\documentclass[conference]{IEEEtran}

\usepackage[T1]{fontenc}
\usepackage{graphicx}
\usepackage{cite}
\usepackage{caption}
\usepackage{subcaption}
\usepackage{epstopdf}
\usepackage{algorithm}
\usepackage{algpseudocode}

\usepackage{overpic}
\usepackage{amssymb}
\usepackage{amsmath}
\usepackage{amsthm}
\usepackage{array}
\usepackage{color}
\usepackage{url}

\usepackage{overpic}
\usepackage{stfloats}
\usepackage{xspace}
\IEEEoverridecommandlockouts

\theoremstyle{plain}
\newtheorem{theorem}{Theorem}

\def\sinc{\mathrm{sinc}}

\def\imagunit{\mathsf{j}} % Imaginary number

\begin{document}

\title{Achieving Beamfocusing via \\ Two Separated Uniform Linear Arrays}

\author{
\IEEEauthorblockN{Alva Kosasih\IEEEauthorrefmark{1}, \"Ozlem Tu\u{g}fe Demir\IEEEauthorrefmark{2}, 
Emil Bj{\"o}rnson\IEEEauthorrefmark{1}
\thanks{This work was supported by the FFL18-0277 grant from the Swedish Foundation for Strategic Research.}}
\IEEEauthorblockA{\IEEEauthorrefmark{1}Department of Computer Science, KTH Royal Institute of Technology, Stockholm, Sweden (\{kosasih,emilbjo\}@kth.se)}
\IEEEauthorblockA{\IEEEauthorrefmark{2}Department of Electrical-Electronics Engineering, TOBB ETU, Ankara, Turkiye  (ozlemtugfedemir@etu.edu.tr)}

}

\maketitle

\begin{abstract}
This paper investigates coordinated beamforming using a modular linear array (MLA), composed of a pair of physically separated uniform linear arrays (ULAs), treated as sub-arrays. We focus on how such setups can give rise to near-field effects in 6G networks without requiring many antennas. Unlike conventional far-field beamforming, near-field beamforming enables simultaneous data service to multiple users at different distances in the same angular direction, offering significant multiplexing gains. We present a detailed analysis, including analytical expressions of the beamwidth and beamdepth for the MLA. Our findings reveal that using the MLA approach, we can remove approximately $36\%$ of the antennas in the ULA while achieving the same level of beamfocusing.
\end{abstract}


\section{Introduction}
\label{S_Intro}

The commercialization of massive multiple-input multiple-output (mMIMO) technology has been a cornerstone for 5G networks, enabling substantial improvements in spectral efficiency and energy efficiency compared to 4G systems  \cite{2019_Björnson_DSP}. As 5G continues to expand, covering $45\%$ of the world’s population by the end of 2023 \cite{Ericsson2024},  focus is shifting to 6G goals like ultra-high throughput, low latency, and advanced services such as wireless sensing and AI \cite{2019_Björnson_DSP,2020_Saad_IEEENet,2024_Pourkabirian_Commag}. To meet these demands, next-generation MIMO technology will incorporate significantly larger arrays and might be called extremely large aperture arrays (ELAAs). By utilizing ELAAs, we can achieve extreme multiplexing capabilities by exploiting the near-field spherical wave features to higher spatial resolution, leading to significant capacity enhancement \cite{2024_Kosasih_Arxiv}.

Beamforming in the near-field turns into \emph{beamfocusing}, where the focal point has a limited size in both angle and depth. This allows users to be spatially multiplexed in both these domains. The beamfocusing feature exists when the propagation distance is less than the Fraunhofer distance $2D^2/\lambda$ \cite{2021_Björnson_Asilomar}, where $D$ and $\lambda$ denote the aperture length of the array and the wavelength, respectively. It is preferable to achieve the beamfocusing feature through a larger aperture array rather than by utilizing higher frequencies, as the latter can lead to coverage issues \cite{2024_Bjornson_Arxiv}.
However, this comes with significant challenges: the excessive hardware and signal processing complexity and cost involved, as a large aperture array requires managing an extremely large number of antennas (e.g., thousands of antennas). 

A practical way to reduce processing complexity is by decreasing the number of antennas while maintaining the aperture area. This can be achieved with modular arrays, which consist of sub-arrays (e.g., 64 antennas) spaced apart to preserve the aperture size \cite{2024_Meng_WCL,2022_Li_CL,2024_Li_TWC}. Modular arrays repurpose existing mMIMO arrays, coordinating and synchronizing sub-arrays separated by a few meters, which is feasible with cable installations. Each sub-array can include a local processor connected to a central unit. This approach provides a scalable path toward cell-free architectures and suggests a new 6G deployment strategy: replacing a single large base station (BS) with multiple smaller panels on rooftops or building facades, such as office windows, all operating as a single BS.

In this paper, we present an in-depth analysis of the beampattern of the modular linear array (MLA), which consists of a pair of uniform linear arrays (ULAs) separated by a few meters.
The beampattern analysis is based on the beamwidth and beamdepth for MLAs, where we derive their analytical expressions. The MLAs show great potential due to their simple structure, while still fully exploiting near-field focusing with significantly fewer antennas than a standard large ULA. In addition, using the developed analytical expressions, we derive conditions such as the required number of antennas to achieve the desired beamfocusing pattern with the MLA.

\section{Beamfocusing with a Single ULA}
\label{Sect_Prelim}

In this section, we elaborate on how to achieve beamfocusing with a ULA receiver and a single-antenna user transmitter. More specifically, we consider a free-space propagation scenario with an isotropic transmit antenna located at $\left(0,0,z \right)$ and a receiver array deployed across $xy$-plane with its center in the origin. We consider a ULA with $N$ antenna elements. Each antenna element is indexed by $n \in \{1, \dots, N \}$. Antenna $n$ is centered at the point $(\bar{x}_n,0,0)$ given by $   \bar{x}_n = \left(n-\frac{N+1}{2}\right)\delta$,
where $\delta$ is the spacing between adjacent antenna elements. Unless otherwise mentioned, we assume $ \delta=\lambda/2$ throughout the paper. Antenna $n$ covers the aperture area of
\begin{align}
    \mathcal{A}_n = \left\{ (x,y,0) : |x-\bar{x}_n|\leq \frac{\delta}{2}, |y|\leq \frac{\delta}{2} \right\}.
\end{align}
The diagonal of the array is  $D_{\rm array} \triangleq \sqrt{( \delta N
)^2 + \delta^2}$ and is referred to as the aperture length of the array. The phase variations of the impinging wave across the array are negligible when the propagation distance exceeds the Fraunhofer array distance, defined as  $d_{FA} =  2D_{\rm array}^2 /\lambda$ \cite{2021_Björnson_Asilomar}.  

The transmitter sends an emitted signal polarized in the $y$-dimension. The electric field at the point $(x,y,0)$ of the receiver aperture becomes \cite[App. A]{2020_Björnson_JCommSoc}
\begin{equation}\label{eq_II_ElectField}
 E (x,y) =  \frac{E_0}{\sqrt{4 \pi}} \frac{\sqrt{z \left( x ^2+z^2\right)}} 
 { r ^{5/4}}   e^{-\imagunit\frac{2\pi}{\lambda} \sqrt{r}},
\end{equation}
where $r =  x^2 + y^2 + z^2$ is the squared Euclidean distance between the transmitter and the considered point, and $E_0$ is proportional to the electric intensity of the transmitted signal.

Using the Fresnel approximation, the electrical field at the point $(x,y)$ can be approximated as \cite{2021_Björnson_Asilomar}
\begin{equation}\label{eq:E_approx}
    E(x,y) \approx \frac{E_0}{\sqrt{4 \pi} z} e^{-\imagunit \frac{2\pi}{\lambda}\left(z + \frac{x^2}{2z} + \frac{y^2}{2z} \right)},
\end{equation}
where a first-order Taylor approximation of the Euclidean distance between the transmitter and the array is used.

 Considering small antennas and continuous matched filtering (MF) with the focal point  $(0, 0, F)$, we can approximate the normalized antenna array gain as \cite[Eq. (20)]{2021_Björnson_Asilomar}
    \begin{align}
      \hat{G}_{{\rm ULA}}  &=\frac{1}{(NA)^2} \times \nonumber\\
      &\left\vert \sum_{n=1}^N \int_{\mathcal{A}_n} e^{+\imagunit\frac{2\pi}{\lambda}\left(\frac{x^2}{2F}+\frac{y^2}{2F}\right)}e^{-\imagunit\frac{2\pi}{\lambda}\left(\frac{x^2}{2z}+\frac{y^2}{2z}\right)}   dx dy\right \vert^2  .\label{eq:BF_ULA} 
    \end{align}
Here, the received power is normalized by the total power over the antenna aperture.
\begin{theorem}\label{Fresnel_Approx_ULA}
When the transmitter is located at $(0,0,z)$ and the MF is focused on $(0, 0, F)$, the  Fresnel approximation of the normalized array gain for a ULA becomes
\begin{equation}\label{eq_III_ApproxGainRect}
\hat{G}_{\rm ULA} =
\frac{\left( {{C}^{2}}\left( \sqrt{ a } \right)+{{S}^{2}}\left( \sqrt{ a } \right) \right) \left( {{C}^{2}}\left( \sqrt{ a } N \right)+{{S}^{2}}\left( \sqrt{ a } N \right) \right)}{( N a )^{2}},
\end{equation}
 where $C(\cdot )$ and $S(\cdot )$ are the Fresnel integrals \cite{1956_Polk_TAP}, $ a = \frac{\lambda}{8{z}_{\rm eff}}$, and $z_{\rm eff} = \frac{Fz}{|F-z|} $.
\end{theorem}
\begin{proof}
The proof is given in Appendix~\ref{App_Fresnel_Approx}
and is inspired by the derivation in \cite[Eq. (22)]{1956_Polk_TAP}.
\end{proof}
The approximation in \eqref{eq_III_ApproxGainRect} is tight when the transmitter is in the (radiative) near-field of the array where there are only phase variations across the array \cite{2023_TWC_BD}.

\section{MLA with Two ULAs}
\label{S_NF_TwoArr}

\begin{figure}
    \centering
     \begin{overpic}[width=\linewidth]{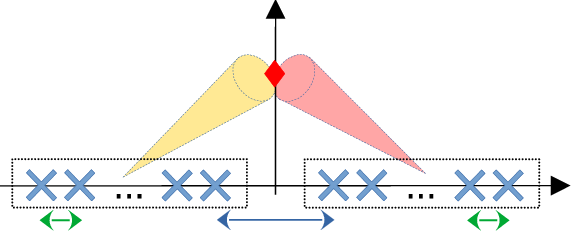}
     \put(3,14){ULA-$1$}
     \put(83,14){ULA-$2$}
     \put(80,-3){\small $ \delta=\lambda/2$}
     \put(5,-3){\small $ \delta=\lambda/2$}
     \put(46.5,-3){\small $\Delta$}
     \put(50,26){\small $(0,0,F)$}
     \put(44,40){\small$z$-axis}
     \put(94,2){\small$x$-axis}
     \end{overpic}
    \vspace{0.7mm}
    \caption{MLA composed by two ULAs deployed $\Delta$\,m apart. Both ULAs are focused at $(0,0,F)$.}
    \label{F_twoULAs}
    \vspace{-3mm}
\end{figure}
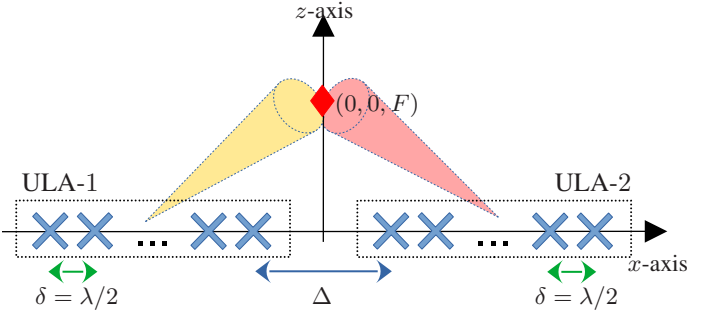

We will now consider an isotropic transmitter located at $(0,0,z)$ communicating with a receiver that is equipped with a pair of ULAs with the centers of the closest antennas of the two ULAs separated by $\Delta$ meters, as illustrated in Fig.~\ref{F_twoULAs}.  
We note that the two ULAs are positioned along the same line on the $x$-axis, which differs from the typical cell-free massive MIMO setup where ULAs can be deployed at arbitrary locations.
Each of the ULAs is equipped with $N$ antennas. The $n$-th antenna element for ULA-$\ell \in \{1,2\}$ is located at the point
\begin{equation}\label{eq:coord_twoULas}
    \bar{x}_n^{(\ell)} = \left( n - \frac{N+1}{2}\right) \delta + {\left( \ell - \frac{3}{2}\right) \left(\Delta + (N-1) \delta\right)}.
\end{equation} 
Notice that the second term in \eqref{eq:coord_twoULas} is either plus or minus   $\overline{\Delta}$, where $\overline{\Delta} = \frac{\Delta + (N-1) \delta}{2}$.
A special case of the above configuration occurs when the separation between the centers of the outermost antennas is $ \delta$ or when $\Delta=0$. In this scenario, the two ULAs will form a single ULA with an inter-element spacing of $ \delta$.
The aperture length of the array now becomes $D_{\rm array} = \sqrt{(\Delta + (2N-1)\delta)^2 + \delta^2}$. Consequently, the Fraunhofer distance of the array increases with the separation  $\Delta$ between the two ULAs. This implies that beamfocusing can be achieved even with a small total number of antennas in the two ULAs by appropriately adjusting the separation  $\Delta$. To this end, we will intensively study the behavior of the beamfocusing for this array model.

We begin by adapting the normalized array gain with MF-beamfocusing  expression in \eqref{eq:BF_ULA} to account for the two ULAs. The expression with the focal point  $(0,0,F)$ is given as 
\begin{align}
\notag& \hat{G}_2 = \frac{1}{{{(2  N \delta^2)}^{2}}}  \\ \notag
& \Bigg|  
\int\limits_{-\frac{N \delta}{2} }^{\frac{N\delta}{2} }
\int\limits_{-\frac{ \delta}{2}}^{\frac{ \delta}{2}} 
e^{\imagunit \frac{2\pi}{\lambda} \left(\frac{(x-\overline{\Delta})^2}{2F} + \frac{y^2}{2F}\right)}
{e}^{-\imagunit\frac{2\pi }{\lambda }\left(\frac{{{(x-\overline{\Delta} )}^{2}}}{2z}+\frac{{{y}^{2}}}{2z} \right)} dxdy  
\\ 
& +  \int\limits_{-\frac{N\delta}{2} }^{\frac{N\delta}{2} }\int\limits_{-\frac{ \delta}{2}}^{\frac{ \delta}{2}} 
e^{\imagunit \frac{2\pi}{\lambda} \left(\frac{(x+\overline{\Delta})^2}{2F} + \frac{y^2}{2F}\right)}
{e}^{-\imagunit\frac{2\pi }{\lambda }\left(\frac{{(x + \overline{\Delta} )^{2}}}{2z}+\frac{{{y}^{2}}}{2z} \right)}  dxdy \Bigg|^2 \label{eq:G2app}.
\end{align}

Let us now compare the normalized array gain of a standard ULA with that of the MLA. For the standard ULA, we consider $N=50$ antennas with half-wavelength spacing. For the MLA, each of the two ULAs consists of $25$ antennas, also half-wavelength spaced, separated by a distance of $\Delta=5$\,m. We set the focus of the array at the location $(0,0,30)$ and the arrays operate at the $15$\,GHz\footnote{The conclusions presented apply to other carrier frequencies as well.} carrier frequency. 
Fig.~\ref{F_BW}(a) depicts the normalized array gain (beampattern) of the ULA accross the $xz$-plane. We  observe that the beam energy spreads out infinitely behind the UE. This is a typical far-field beampattern and was expected since the UE distance is larger than the Fraunhofer array distance of $23$\,m. 
Conversely, Fig.~\ref{F_BW}(b) shows the MLA beampattern, where the signal is focused at the same UE location. This demonstrates that beamfocusing can be achieved using the MLA configuration, even with the same total number of antennas as in the ULA. In the following, we will provide a theoretical analysis of the MLA beampattern.

\begin{figure}
\centering
\subfloat[A single ULA.]
{\includegraphics[width=0.508\textwidth]{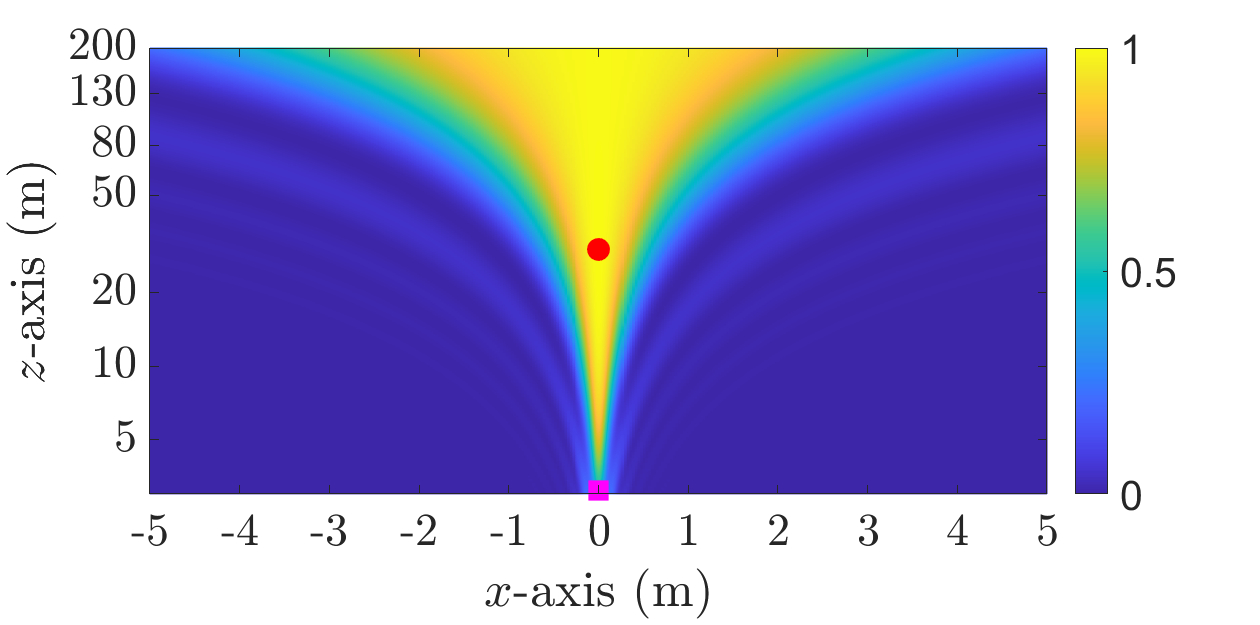}}\hfill
\centering
\subfloat[Two ULAs separated by $5$\,m apart.]
{\includegraphics[width=0.51\textwidth]{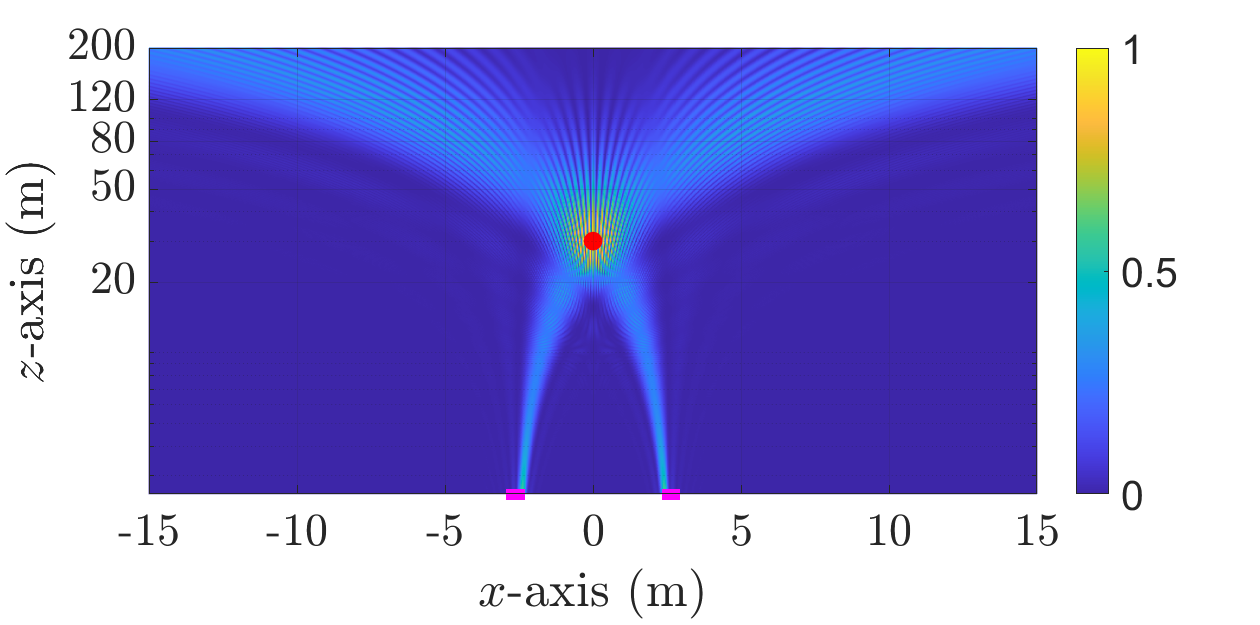}}
\caption{The beampattern of a single ULA and MLA. The magenta rectangles indicate the location of the ULA while the red circle indicates the location of the UE.}
\label{F_BW}
\end{figure}

 \subsection{Beamwidth}\label{S_Beamwidth_2arrays}

The  transverse beamwidth (BW) at the focal point $(0,0,F)$ can be analyzed based on the normalized gain from the transmitter along the $x$-axis.
We assume that the transmitter can be at any point of $(x_t,0,z=F)$. Accordingly, we can write the normalized array gain approximation in \eqref{eq:G2app} as in \eqref{eq:BW_twoArr} at the top of the next page. Note that the transmitter can now be located at any point along the $x$-axis, corresponding to a non-broadside transmission. Consequently, the transmitter's coordinate along the $x$-axis must be incorporated into the phase-shift expression in the normalized array gain. 
\begin{figure*}
\begin{align}
\notag
    &\frac{1}{{{(2  N \delta^2)}^{2}}}  \left |
      \left(
      \int\limits_{-\frac{N\delta}{2} }^{\frac{N\delta}{2} }
      \int\limits_{-\frac{ \delta}{2}}^{\frac{ \delta}{2}} 
       e^{\imagunit \frac{2\pi}{\lambda} \left(\frac{(x-\overline{\Delta})^2}{2F} + \frac{y^2}{2F}\right)}
     {e}^{-\imagunit\frac{2\pi }{\lambda }\left(\frac{{{(x-\overline{\Delta} -x_t)}^{2}}}{2z}+\frac{{{y}^{2}}}{2z} \right)} dx
    dy  +  \int\limits_{-\frac{N\delta}{2} }^{\frac{N\delta}{2} }\int\limits_{-\frac{ \delta}{2}}^{\frac{ \delta}{2}} 
     e^{\imagunit \frac{2\pi}{\lambda} \left(\frac{(x+\overline{\Delta})^2}{2F} + \frac{y^2}{2F}\right)}
     {e}^{-\imagunit\frac{2\pi }{\lambda }\left(\frac{{(x + \overline{\Delta} -x_t)^{2}}}{2z}+\frac{{{y}^{2}}}{2z} \right)}  dxdy \right) \right|^{2}\nonumber 
    \\
    &=\frac{1}{{{(2  N \delta^2)}^{2}}}  \left|
      \delta  \left(\int\limits_{-\frac{N\delta}{2}  -\overline{\Delta}}^{\frac{N\delta}{2}  -\overline{\Delta}}
       e^{\imagunit \frac{2\pi}{\lambda z} u_1 x_t} du_1 + 
     \int\limits_{-\frac{N\delta}{2}   +\overline{\Delta}}^{\frac{N\delta}{2}   +\overline{\Delta}} e^{\imagunit \frac{2\pi}{\lambda z} u_2 x_t} du_2
      \right) \right|^{2} \nonumber
    \\
    &=\Bigg|\frac{1}{2} \sinc\left(\frac{N x_t}{2z}\right) \left(e^{\imagunit \left(\frac{2\pi \overline{\Delta} x_t}{\lambda z}\right) } +   e^{-\imagunit \left(\frac{2\pi \overline{\Delta} x_t}{\lambda z}\right) } \right)  \Bigg|^2=\Bigg|\sinc\left(\frac{N x_t}{2z}\right) \cos\left(\frac{2 \pi \overline{\Delta} x_t}{\lambda z} \right)  \Bigg|^2   \leq \sinc^2\left(\frac{N x_t}{2z}\right).\label{eq:BW_twoArr}
\end{align}
\hrulefill \vspace{-4mm}
\end{figure*}

\begin{figure}
\centering
\subfloat[Total aperture length of $1$\,m.  ]
{ \includegraphics[width=0.48\textwidth]{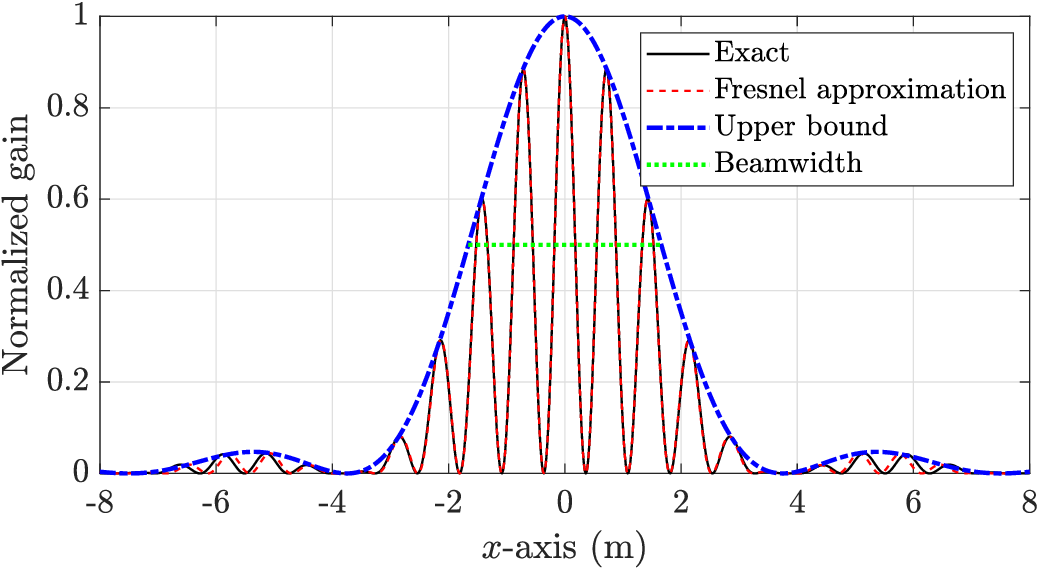}}\hfill
\centering
\subfloat[Total aperture length of $2$\,m.]
{ \includegraphics[width=0.48\textwidth]{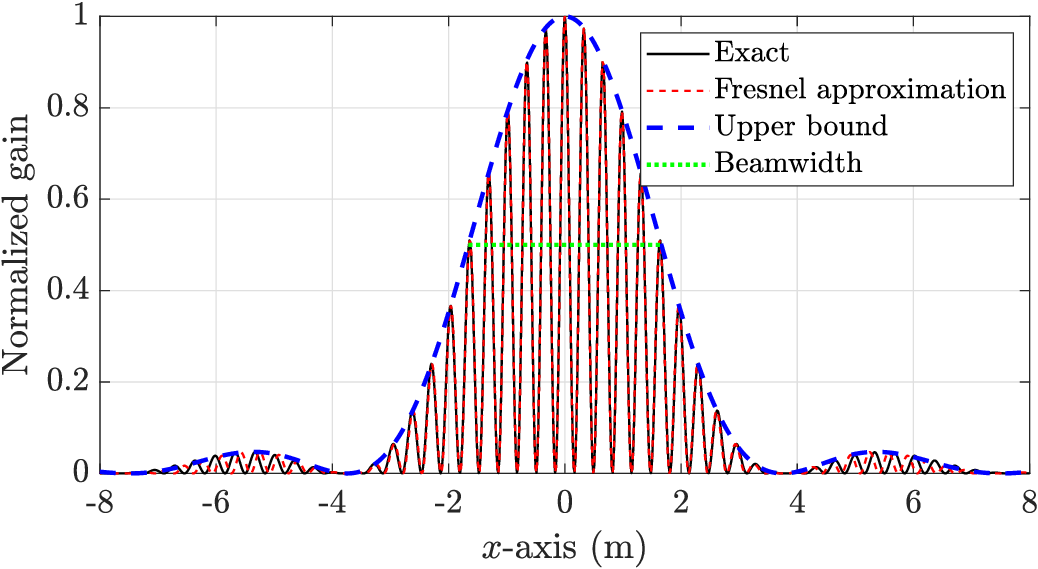}}
\caption{The beampattern of MLA along $x$-axis.} \vspace{-3mm}
\label{F_BW2}
\end{figure}

We evaluate the beamwidth of the MLA in Fig.~\ref{F_BW2}. Each ULA in MLA is equipped with $N=16$ antennas. The focal point is at $(0,0,30)$\,m.
The black curve corresponds to the exact normalized array gain, the red-dashed curve corresponds to the Fresnel approximation of the normalized array gain obtained from \eqref{eq:BW_twoArr}, the blue-dashed curve corresponds to the upper bound of the normalized array gain expression in \eqref{eq:BW_twoArr}, and the green-dashed line corresponds to the width of the upper bound curve which we refer to as the beamwidth.

In Fig.~\ref{F_BW2}(a), the two ULAs are separated by a distance of $1$\,m. We observe that the normalized array gain oscillates along the $x$-axis. The gain oscillates more when we consider a larger separation between the two ULAs (i.e., $2$\,m), illustrated in Fig.~\ref{F_BW2}(b). The gain fluctuations are undesirable because they imply that the beamforming gain can drop to zero unexpectedly, even when the UE is within the beamfocusing region. This region is defined as the area where the upper bound is greater than or equal to half of its maximum value ($0.5$), marked by the green line, corresponding to its $3$\,dB beamwidth. 
This beamwidth can be analytically expressed as 
\begin{equation}\label{eq:3dB_BW}
    0.443 \approx \frac{N |x_t|}{2F} \rightarrow {\rm BW_{\rm 3 dB}} \approx \frac{1.77F}{N}, 
\end{equation}
since  $\sinc^2(0.443)\approx 0.5$. This value for the upper bound does not depend on the inter-array spacing $\Delta$. However, the spacing will determine the oscillations in the actual gain pattern.

Since gain oscillations within the beamfocusing region are undesirable, we will now analytically derive the necessary conditions to eliminate them. 
We first characterize the nulls of the small beams within the beamfocusing region. The nulls appear when $\cos(\frac{2 \pi \overline{\Delta}x_t}{\lambda F}) = 0$. Since $\cos\left((2 k +1 )  \pi/2\right)=0$ for any integer value of $k$, we obtain that the nulls occur when 
\begin{equation}
    x_t = \pm\frac{\lambda F}{4 \overline{\Delta}}(2k+1).
\end{equation}
We can then calculate the length of two adjacent nulls as $\frac{\lambda F}{2\overline{\Delta}}$, obtained by setting $k=0$.
Therefore, the number of nulls within the beamfocusing region is approximately 
\begin{equation}\label{eq:num_nulls}
     \left\lfloor\frac{{\rm BW_{3dB}}} {\lambda F/ (2  \overline{\Delta})}  \right\rfloor +1.
\end{equation}
From \eqref{eq:num_nulls}, we understand that the number of small beams within the beamfocusing region decreases when the number of antennas per ULA $(N)$ increases.

Since we do not want any fluctuation inside of the beamfocusing region, we want to have only a single beam inside of it. The condition that satisfies this is $\left\lfloor \frac{3.54 \overline{\Delta}} {{\lambda N  }} \right\rfloor \leq 1$, which can also be written as $ \frac{3.54 \overline{\Delta}} {{\lambda N  }}  < 2$ .
Therefore,  the number of antennas in each ULA that satisfies the condition is
\begin{equation}\label{eq:Nmin}
    N > {\frac{ 1.77 \overline{\Delta}}{\lambda}}.
\end{equation}
This  implies that the ratio between the blank space $\Delta$ and filled antennas $N\lambda$  in the whole aperture of the MLA must be less than $57 \%$ to get only a small beam inside of the beamfocusing region. To illustrate this, we plot the beampattern across the $xz$-plane and the normalized array gain along the $x$-axis in Fig.~\ref{F_Beampattern_nice}. We set the total aperture length of the MLA as $D_{\rm array} = 2$\,m, where the separation between the centers of the outermost antennas is $\Delta = 0.72$\,m. Hence, the ratio of $\Delta/(N\lambda)$ is around $56 \%$.
In the other perspective, if the aperture length of the array is $2$\,m, then, there are $200$ antennas in the case of ULA, however, only $128$ antennas with the MLA configuration are needed to achieve a desired beamfocusing.
The beam can now be focused nicely within the desired location $(0,0,30)$ in contrast to the beampattern shown in Fig.~\ref{F_BW}(b) which contains many ripples. The beampattern in  Fig.~\ref{F_BW}(b) was generated using MLA with the same aperture length of $2$\,m but smaller number of antennas $N=16$.

\begin{figure}
\centering
\subfloat[Beampattern.]
{\includegraphics[width=\linewidth]{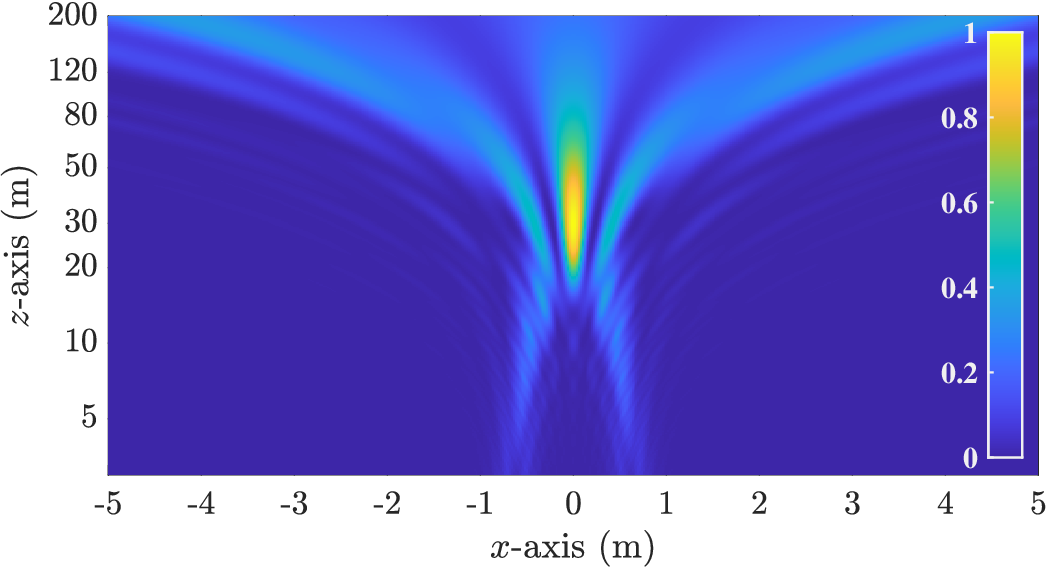}}\hfill
\centering
\subfloat[Normalized array gain.]
{\includegraphics[width=\linewidth]{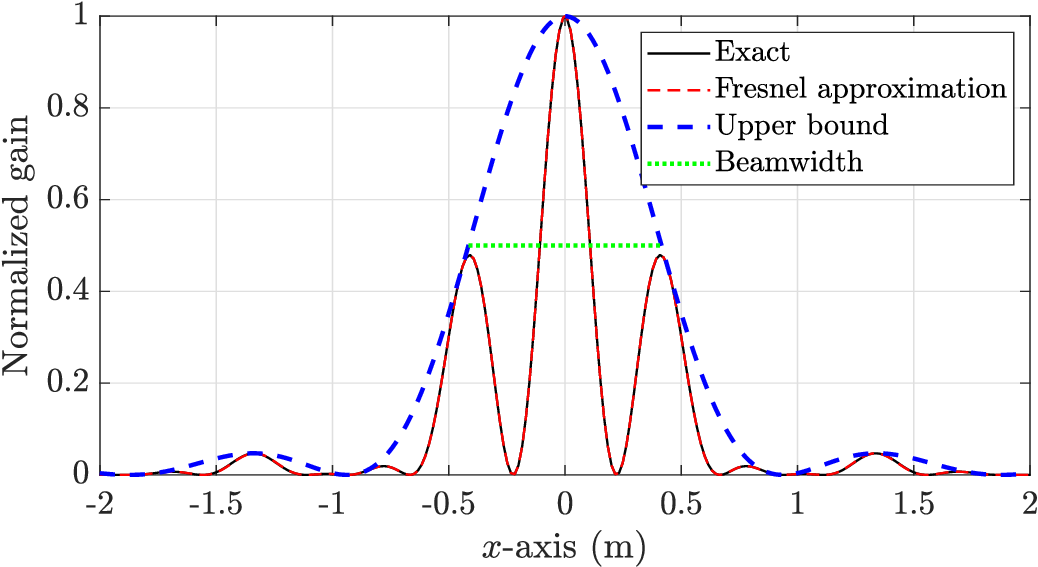}}
\caption{MLA beamforming with the number of antennas in each ULA $N=64$ and focal point $(0,0,30)$\,m.}
\label{F_Beampattern_nice}
\end{figure}

 \subsection{Beamdepth}

 The depth of the beam can be analyzed based on the normalized array gain, where we expect to see the energy focused on a point along the $z$-axis. We now consider a transmitter at $(0,0,z)$ and the MF is focused at point $(0,0,F)$. We can derive the closed-form expression of  \eqref{eq:G2app}, stated in the following theorem. 
 
\begin{theorem}\label{Fresnel_Approx_twoULAs}
The  Fresnel approximation of the normalized array gain for MLA when the transmitter is located at $(0,0,z)$ and the MF is focused on $(0, 0, F)$ is
\begin{multline}\label{eq_III_ApproxGainTwoULAs}
\hat{G}_{2} =  \frac{1}{(2 N a)^2}  \left( C^2(\sqrt{a}) + S^2(\sqrt{a}) \right)  \\
\cdot\left( \left( C(\beta_1) +C(\beta_2) \right)^2 + \left( S(\beta_1) +S(\beta_2) \right)^2\right),
\end{multline}
 where   $\beta_1 = \sqrt{a} N + \sqrt{\frac{2 }{\lambda z_{\rm eff}}}\overline{\Delta}$,  $\beta_2 = \sqrt{a} N - \sqrt{\frac{2 }{\lambda z_{\rm eff}}}\overline{\Delta}$, $a = \frac{\lambda}{8{z}_{\rm eff}}$, and $z_{\rm eff} = \frac{Fz}{|F-z|} $.
\end{theorem}
\begin{proof}
The proof is given in Appendix~\ref{App_Fresnel_Approx_twoULAs}.
\end{proof}

\begin{figure}
    \centering
    \includegraphics[width=\linewidth]{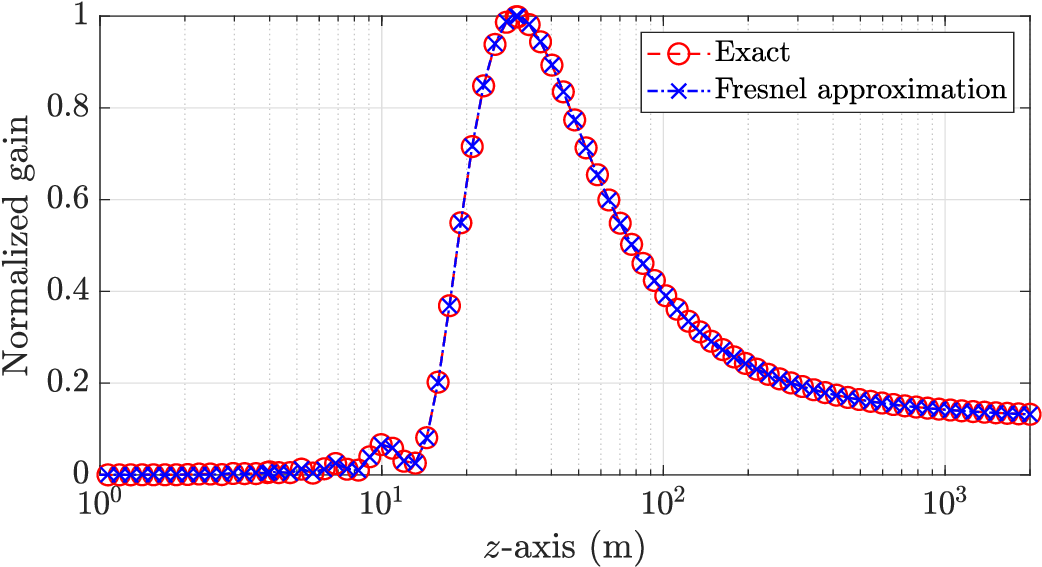}
    \caption{The Fresnel approximation of MLA with the number of antennas in each ULA $N=64$ and focal point $(0,0,30)$\,m.}
    \label{F_BD1} \vspace{-3mm}
\end{figure}

 We evaluate the Fresnel approximation of the normalized array gain in \eqref{eq_III_ApproxGainTwoULAs} by plotting Fig.~\ref{F_BD1} for the same setup as in Fig.~\ref{F_Beampattern_nice}. The Fresnel approximation results in virtually the same gain as the exact value obtained by numerically evaluating the array gain expression with MF at the focal point $(0,0,30)$\,m.

\section{Conclusions}
We presented a comprehensive analysis of near-field beamforming using MLAs, showcasing their advantages in achieving beamfocusing. 
We demonstrated that, with the same number of antennas, beamfocusing is not achievable in a ULA configuration, but it can be achieved with the MLA.
We derived the analytical expressions for beamwidth and beamdepth providing insights into designing the MLA to enable beamfocusing with not so large number of antennas. 
The simulation result confirmed the findings, which shows that only about $64\%$ of the antennas in the ULA are needed to achieve beamfocusing when using the MLA.

\appendices
\section{Proof of Theorem~\ref{Fresnel_Approx_ULA}}
\label{App_Fresnel_Approx}
By defining $z_{\rm eff} = \frac{Fz}{|F-z|}$, we can rewrite \eqref{eq:BF_ULA} as
\begin{align}\label{G-derivation}
    \hat{G}_{\rm ULA} =\frac{1}{(NA)^{2}}{{\left| \int\limits_{-\frac{N\delta}{2} }^{\frac{N\delta}{2} }\int\limits_{-\frac{\delta}{2}}^{\frac{\delta}{2}} 
    {e}^{-\imagunit\frac{\pi }{\lambda } \frac{ {{x}^{2}}}{z_{\rm eff}}  }  {e}^{-\imagunit\frac{\pi }{\lambda } \frac{ {{y}^{2}}}{z_{\rm eff}}  }dydx \right|}^{2}}.
\end{align}
The evaluation of the anti-derivatives in \eqref{G-derivation} yields  \cite{1956_Polk_TAP}
\begin{equation}
\hat{G}_{\rm ULA} =
\frac{\left( {{C}^{2}}\left( \sqrt{ a } \right)+{{S}^{2}}\left( \sqrt{ a } \right) \right) \left( {{C}^{2}}\left( \sqrt{ a } N \right)+{{S}^{2}}\left( \sqrt{ a } N \right) \right)}{( N a )^{2}},
\end{equation}
where $C\left(\cdot \right)$ and $S\left(\cdot \right)$ are the Fresnel integrals, and $ a = \frac{\lambda}{8{z}_{\rm eff}}$. This completes the proof.

\section{Proof of Theorem~\ref{Fresnel_Approx_twoULAs}}
\label{App_Fresnel_Approx_twoULAs}

 We write \eqref{eq:G2app} as
\begin{align}
\notag
& \hat{G}_2= \frac{1}{{{(2  N \delta^2)}^{2}}}  \Bigg| \int\limits_{-\frac{ \delta}{2}}^{\frac{ \delta}{2}} 
{e}^{\imagunit\frac{\pi }{\lambda }\frac{{{y}^{2}}}{z_{\rm eff}} } dy \\  
&\hspace{2cm}   \cdot\left( \int\limits_{-\frac{N\delta}{2}}^{\frac{N\delta}{2} }
e^{\imagunit \frac{\pi}{\lambda} \frac{(x-\overline{\Delta})^2}{ z_{\rm eff}} } dx +
\int\limits_{-\frac{N\delta}{2} }^{\frac{N\delta}{2} }
e^{\imagunit \frac{\pi}{\lambda} \frac{(x+\overline{\Delta})^2}{ z_{\rm eff}} }
 dx \right) \Bigg|^2. \label{G2a}
\end{align} 
The evaluation of the anti-derivatives in \eqref{G2a} using the Fresnel integrals yields
$\hat{G}_{2} = \frac{1}{{{(4  N a)}^{2}}}  
 \left( C^2(\sqrt{a}) + S^2(\sqrt{a}) \right)  
 \cdot( ( C(\beta_1) - C(-\beta_1) +C(\beta_2) - C(-\beta_2) )^2$ 
 $+ 
 ( S(\beta_1) -S(-\beta_1) + S(\beta_2) -S(-\beta_2) )^2)$,
where  $ a = \frac{\lambda}{8{z}_{\rm eff}}$, $\beta_1 = \sqrt{a} N + \sqrt{\frac{2 }{\lambda z_{\rm eff}}}\overline{\Delta}$, and  $\beta_2 = \sqrt{a} N - \sqrt{\frac{2 }{\lambda z_{\rm eff}}}\overline{\Delta}$. Since $C(\cdot)$ and $S(\cdot)$ are both odd functions, we can simplify $\hat{G}_{2}$ as
$\hat{G}_{2} = \frac{1}{{{(2  N a)}^{2}}}  
 ( C^2(\sqrt{a}) + S^2(\sqrt{a}) 
\cdot( \left( C(\beta_1) +C(\beta_2) \right)^2 + \left( S(\beta_1)  + S(\beta_2)  \right)^2)$,
which completes the proof.

\bibliographystyle{IEEEtran}
\bibliography{IEEEabrv,refs}

\end{document}